\documentclass[11pt,reqno, final]{amsart}

\usepackage{color}
\usepackage[colorlinks=false]{hyperref}
\usepackage{amsmath, amssymb, amsthm}
\usepackage{mathrsfs}
\usepackage{mathtools}
\usepackage[noabbrev,capitalize,nameinlink]{cleveref}
\crefname{equation}{}{}
\usepackage{fullpage}
\usepackage{graphics}
\usepackage{pifont}
\usepackage{tikz}
\usepackage{bbm}
\usepackage[T1]{fontenc}
\usetikzlibrary{arrows.meta}

\usepackage{environ}
\usepackage{framed}
\usepackage{url}
\usepackage[linesnumbered,ruled,vlined]{algorithm2e}
\usepackage[noend]{algpseudocode}
\usepackage[labelfont=bf]{caption}
\usepackage{framed}
\usepackage[framemethod=tikz]{mdframed}
\usepackage{appendix}
\usepackage{graphicx}
\usepackage[textsize=tiny]{todonotes}
\usepackage{tcolorbox}
\allowdisplaybreaks

\crefname{algocf}{Algorithm}{Algorithms}

\crefname{equation}{}{} %remove ``Equation''
\AtBeginEnvironment{appendices}{\crefalias{section}{appendix}} %appendices

\usepackage[color]{showkeys} %add in 'final' into parameter to remove showkeys

\colorlet{refkey}{orange!20}
\colorlet{labelkey}{blue!30}

\crefname{algocf}{Algorithm}{Algorithms}

\numberwithin{equation}{section}
\newtheorem{theorem}{Theorem}[section]

\newtheorem{lemma}[theorem]{Lemma}

\crefname{claim}{Claim}{Claims}

\newtheorem{corollary}[theorem]{Corollary}

\newtheorem*{question*}{Question}

\theoremstyle{definition}
\newtheorem{definition}[theorem]{Definition}

\newtheorem*{definition*}{Definition}

\theoremstyle{remark}
\newtheorem*{remark}{Remark}

\newcommand{\snorm}[1]{\lVert#1\rVert}

\newcommand{\sang}[1]{\langle #1 \rangle}

\newcommand{\mb}{\mathbb}

\newcommand{\mc}{\mathcal}
\newcommand{\mf}{\mathfrak}

\newcommand{\on}{\operatorname}
\newcommand{\tsc}{\textsc}

\allowdisplaybreaks

\title{A Gaussian fixed point random walk}

\author[A1]{Yang P. Liu}
\address{Department of Mathematics, Stanford University,
Stanford, CA 94305, USA}
\email{yangpliu@stanford.edu}

\author[A2]{Ashwin Sah}
\author[A3]{Mehtaab Sawhney}
\address{Department of Mathematics, Massachusetts Institute of Technology, Cambridge, MA 02139, USA}
\email{\{asah,msawhney\}@mit.edu}

\begin{document}
\begin{abstract}
In this note, we design a discrete random walk on the real line which takes steps $0, \pm 1$ (and one with steps in $\{\pm 1, 2\}$) where at least $96\%$ of the signs are $\pm 1$ in expectation, and which has $\mc{N}(0,1)$ as a stationary distribution. As an immediate corollary, we obtain an online version of Banaszczyk's discrepancy result for partial colorings and $\pm 1, 2$ signings. Additionally, we recover linear time algorithms for logarithmic bounds for the Koml\'{o}s conjecture in an oblivious online setting.
\end{abstract}

\maketitle

\section{Introduction}\label{sec:introduction}
In the (oblivious) online vector discrepancy problem an adversary fixes vectors $\{v_i\}_{i\in [t]}$ in advance and the objective is to assign signs $\epsilon_i\in \{-1,1\}$ based only on vectors $v_1,\ldots,v_i$ to maintain that $\snorm{\sum_{i\le t'} \epsilon_iv_i}_\infty$ is small at all times $t' \in [t]$. 
Vector balancing includes a number of different problems in discrepancy theory including Spencer's \cite{Spe85} work on set discrepancy. Spencer's ``six standard deviations suffice'' result states that given vectors $v_1,\ldots,v_n\in \{0,1\}^n$ there exists a $\pm 1$-signing such that $\snorm{\sum_{i\le n} \epsilon_iv_i}_{\infty}\le 6\sqrt{n}$. Conjecturally, however, the restriction to $\{0,1\}^n$ vectors can be relaxed to a norm condition. In particular, the Koml\'{o}s conjecture states that given $v_1,\ldots,v_t$, each of at most unit length, there exists a sequence of signs $\epsilon_1,\ldots,\epsilon_t$ such that $\snorm{\sum_{i\le t} \epsilon_iv_i}_{\infty} = O(1)$. Despite substantial effort, the Koml\'{o}s conjecture is still open and the best known bounds due to Banaszczyk \cite{Ban98} give the existence of a sequence of signs so that $\epsilon_1,\ldots,\epsilon_t$ such that $\snorm{\sum_{i\le t} \epsilon_iv_i}_{\infty} = O(\sqrt{\min(\log n,\log t}))$. However, these original proofs were by their nature non-algorithmic.

More recent research in theoretical computer science has focused on developing algorithmic versions of these results starting with the Bansal \cite{Ban10} and Lovett-Meka \cite{LM15} polynomial-time algorithms for Spencer's \cite{Spe85} ``six standard deviations suffice''. Since then, there have been several other constructive discrepancy minimization algorithms \cite{Rot17,ES18,BDG16,BG17,BDGL18,DNTT18}. Notably for our purposes, Bansal, Dadush, Garg \cite{BDG16} and Bansal, Dadush, Garg, Lovett \cite{BDGL18} have made the work of Banaszczyk \cite{Ban98} algorithmic. However in all cases these algorithms require all vectors to be known at the start and hence do not extend to the online setting.

In the online setting, significant work has been devoted to the case where $v_i$ are drawn from a fixed (and known) distribution $\mf{p}$ supported on $[-1,1]^n$. In the setting where $\mf{p}$ is uniform on $[-1,1]^n$, Bansal and Spencer \cite{BS19} showed one can maintain $\max_{t'\le t}\snorm{\sum_{i\le t'} \epsilon_iv_i}_\infty\le O(\sqrt{n} \log t)$. In the more general setting where $\mf{p}$ is a general distribution supported on $[-1,1]^n$, Aru, Narayanan, Scott, and Venkatesan \cite{ANSV18} achieved a bound of $O_n(\sqrt{\log t})$ (where the implicit dependence on $n$ is super-exponential) and Bansal, Jiang, Meka, Singla, and Sinha \cite{BJMSS20} (building on work of Bansal, Jiang, Singla, and Sinha \cite{BJSS20}) achieved an $\ell_{\infty}$ guarantee of $O(\sqrt{n}\log(nt)^4)$.

In this work we focus on the online setting where the only guarantee is $\snorm{v_i}_2\le 1$. The only previous work in this oblivious online setting is the following result of Alweiss, the first author, and the third author \cite{ALS20}.

\begin{theorem}[{\cite[Theorem~1.1]{ALS20}}]
\label{thm:balance}
For any vectors $v_1, v_2, \cdots, v_t \in \mb{R}^n$ with $\|v_i\|_2 \le 1$ for all $i \in [t]$, there exists an online algorithm $\tsc{Balance}(v_1, \cdots, v_t, \delta)$ which maintains $\snorm{\sum_{i\le t'} \epsilon_iv_i}_\infty = O\left(\log(nt/\delta)\right)$ for all $t' \in [t]$ with probability at least $1-\delta$. 
\end{theorem}

The proof in \cite{ALS20} relies on a coupling procedure which compares the distribution of $\sum_{i\le t}\epsilon_i v_i$ to a Gaussian at each stage via a stochastic domination argument and then deduces the necessary tail bounds. In this work, we recover \cref{thm:balance} (in fact with a slightly improved dependence) as well as the following corollary.

\begin{corollary}\label{cor:balance-new}
For any vectors $v_1, v_2, \cdots, v_t \in \mb{R}^n$ with $\|v_i\|_2 \le 1$ for all $i \in [t]$, there exists an online algorithm which assigns $\epsilon_i\in \{\pm 1, 2\}$ and maintains $\snorm{\sum_{i\le t'} \epsilon_iv_i}_\infty = O(\sqrt{\log(nt/\delta)})$ for all $t' \in [t]$ with probability at least $1-\delta$. 
\end{corollary}

This result essentially recovers the best known bound on the Koml\'os conjecture due to Banaszczyk \cite{Ban98} in an online algorithmic fashion, with the slight defect of requiring a $+2$-signing option. Furthermore due to the online nature of the algorithm, the algorithm will run in essentially input-sparsity time which is substantially faster than the Gram-Schmidt walk \cite{BDGL18} which gives an algorithmic proof of the result of \cite{Ban98} (without the defect of requiring a $+2$-signing option).

Our results are based on the observation that there exists Markov chains on $\mb{R}$ with transition steps of $0,\pm 1$ or $\pm 1, 2$ such that $\mc{N}(0,1)$ is a stationary distribution (as well as $\mc{N}(0,\sigma^2)$ for appropriate values of $\sigma$). Note that no such walk exists for $\pm 1$ steps as $\sum_{n\in \mb{Z}} (-1)^ne^{-n^2/2}\neq 0$ and therefore any $\pm 1$ walk fails the natural ``parity constraint'' that the total mass on even integers is mapped to the odd integers and vice versa under one step.

The remainder of the paper is organized as follows. In \cref{sec:0-walk} we construct the required Markov chain on $\mb{R}$ with transition steps of $0,\pm 1$ such that $\mc{N}(0,\sigma^2)$ is a stationary distribution. In \cref{sec:2-walk} we extend this to a walk with transition steps of $\pm 1, 2$ as long as $\sigma\ge 1$. Finally, in \cref{sec:algorithms} we deduce the various algorithmic consequences.

\subsection{Notation}
Throughout this paper let $\mc{N}(\mu,\sigma^2)$ denote the Gaussian random variable with mean $\mu$ and variance $\sigma^2$. Furthermore, let $\on{nnz}(\{v_i\}_{i\in S})$ denote the total number of non-zero entries of the vectors $\{v_i\}_{i\in S}$. 

\section{\texorpdfstring{$0, \pm 1$}{0, +-1} walk}\label{sec:0-walk}
\begin{definition}\label{def:jacobi-walk}
Given $\sigma > 0$ and $f\in[-1/2,1/2]$, consider the following random walk on $f+\mb{Z}$. For $n\ge 1$ the state $n+f$ moves to $n+1+f$ with probability $p_\sigma(n+f)$ and to $n-1+f$ otherwise, and the state $-n+f$ moves to $-n-1+f$ with probability $p_\sigma(n-f)$ and to $-n+1+f$ otherwise. Finally, the state $f$ moves to $1+f$ with probability $p_\sigma(f)$, to state $-1+f$ with probability $p_\sigma(-f)$, and stays at $f$ with probability $r_\sigma(f)$. Here
\begin{align*}
p_\sigma(x) &= \sum_{j\ge 1}(-1)^{j-1}\exp\bigg(-\frac{j^2+2xj}{2\sigma^2}\bigg)\\
r_\sigma(f) &= \sum_{j=-\infty}^\infty (-1)^j\exp\bigg(-\frac{j^2+2fj}{2\sigma^2}\bigg)
\end{align*}
for all $x\in\mb{R}$.
\end{definition}

These series clearly absolutely converge. We prove that these indeed correspond to consistent probabilities giving a walk, and additionally show that this walk preserves the discrete Gaussian distribution on $f+\mb{Z}$ (i.e., $\mc{N}(0,\sigma^2)|_{f+\mb{Z}}$).
\begin{lemma}\label{lem:validity}
For $\sigma > 0$ and $f\in[-1/2,1/2]$, we have that $p_\sigma(n\pm f)\in(0,1)$ for all $n\ge 0$, that $p_\sigma(f)+r_\sigma(f)+p_\sigma(-f) = 1$, that $r_\sigma(f)\in[0,1]$, and that furthermore
\[r_\sigma(f)\le e^{-\sigma^2}\]
if $\sigma\ge 1/2$. Additionally, $\mc{N}(0,\sigma^2)|_{f+\mb{Z}}$ is stationary under a step of random walk defined in \cref{def:jacobi-walk} with parameters $\sigma,f$.
\end{lemma}
\begin{proof}
First, note that $\exp(-(j^2+2xj)/\sigma^2)$ is strictly decreasing on integers $j\ge 1$ as long as $x\ge -1/2$. Therefore $p_\sigma(x)$ is given by an alternating series with strictly decreasing terms, and we immediately deduce
\[0 < p_\sigma(x)\le\exp\bigg(-\frac{j^2+2xj}{2\sigma^2}\bigg) < 1.\]
Since $n+f,n-f\ge -1/2$ for $n\ge 0$, we see that $p_\sigma(n\pm f)\in(0,1)$, as desired. Second, note that
\[p_\sigma(-f)+r_\sigma(f)+p_\sigma(f) = 1\]
holds as trivially everything except the $j = 0$ term of the sum for $r_\sigma(f)$ cancels. Third, we have for $u = \exp(-1/(2\sigma^2))$ and $v = \sqrt{-1}\exp(-f/(2\sigma^2))$ that $|u| < 1$ and $v\neq 0$, hence the Jacobi triple product identity (see \cite{And65} for a short but slick proof) yields
\begin{align}
r_\sigma(f) = \sum_{j=-\infty}^\infty u^{j^2}v^{2j} &= \prod_{j=1}^\infty(1-u^{2j})(1+u^{2j-1}v^2)(1+u^{2j-1}v^{-2})\notag\\
&= \prod_{j=1}^\infty(1-e^{-j/\sigma^2})(1-e^{-(2j+2f-1)/(2\sigma^2)})(1-e^{-(2j-2f-1)/(2\sigma^2)}).\label{eq:jacobi-triple-product}
\end{align}
Since $f\in[-1/2,1/2]$ we see each term is nonnegative and clearly less than $1$, so $r_\sigma(f)\in[0,1]$ is immediate. Therefore we indeed have a well-defined walk. In fact, we see that
\[r_\sigma(f)\le r_\sigma(0)\le \prod_{j=1}^\infty (1-e^{-j/\sigma^2})^3 \le  \prod_{j=1}^{\lfloor\sigma^2\rfloor}(1-e^{-j/\sigma^2})^{3}\le (1-e^{-1})^{3\lfloor\sigma^2\rfloor}.\]
This is at most $\exp(-\sigma^2)$ for $\sigma\ge 2$, and we can further numerically check that $r_\sigma(0)\le\exp(-\sigma^2)$ for $\sigma\in[1/2,2]$.

Now we show that this walk preserves $\mc{N}(0,\sigma^2)|_{f+\mb{Z}}$. Note that
\[1-p_\sigma(x) = \sum_{j\ge 0}(-1)^j\exp\bigg(-\frac{j^2+2xj}{2\sigma^2}\bigg).\]
Therefore
\begin{align*}
p_\sigma(x-1)&\exp\bigg(-\frac{(x-1)^2}{2\sigma^2}\bigg) + (1-p_\sigma(x+1))\exp\bigg(-\frac{(x+1)^2}{2\sigma^2}\bigg)\\
&= \sum_{j\ge 1}(-1)^{j-1}\exp\bigg(-\frac{(j+x-1)^2}{2\sigma^2}\bigg) + \sum_{j\ge 0}(-1)^j\exp\bigg(-\frac{(j+x+1)^2}{2\sigma^2}\bigg)\\
&= \exp\bigg(-\frac{x^2}{2\sigma^2}\bigg).
\end{align*}
Since the pdf of $\mc{N}(0,\sigma^2)|_{f+\mb{Z}}$ at $n+f$ is proportional to $\exp(-(n+f)^2/(2\sigma^2))$, we find that the random walk preserves this distribution at $n+f$ for all $n\neq 0$ (applying the above equation at values $x = n\pm f$). Furthermore, the final distribution is clearly still supported on $f+\mb{Z}$, therefore the probability at $n = 0$ is also preserved as the total sum is $1$.
\end{proof}

We immediately derive a walk which preserves $\mc{N}(0,\sigma^2)$ by piecing together all $f\in[-1/2,1/2)$. Let $J_x^\sigma$ be the random variable defined by writing $x = n+f$, where $f\in[-1/2,1/2)$, and then performing a step according to \cref{def:jacobi-walk}.
\begin{lemma}\label{lem:gaussian-jacobi}
If $Z = \mc{N}(0,\sigma^2)$ then $Z + J_Z^\sigma$ is distributed as $\mc{N}(0,\sigma^2)$.
\end{lemma}

\section{\texorpdfstring{$\pm 1, 2$}{+-1, 2} walk}\label{sec:2-walk}
We now consider a variant of the above random walk with discrete $\pm 1$ and $2$ steps. Recall the definition of $p_\sigma(x)$ and $r_\sigma(f)$ from earlier. We will require the following numerical estimate which is deferred to \cref{sec:appendix}.
\begin{lemma}\label{lem:inequality}
If $\sigma\ge 1$ and $f\in[-1/2,1/2]$ then
\[p_\sigma(1+f)\ge r_\sigma(f)\exp\bigg(\frac{2f+1}{2\sigma^2}\bigg).\]
\end{lemma}
\begin{remark}
This inequality is immediate for large $\sigma$ as the left uniformly tends to $1/2$ and the right uniformly decays to zero.
\end{remark}

\begin{definition}\label{def:ramanujan-walk}
Given $\sigma\ge 1$ and $f\in[-1/2,1/2]$, consider the following random walk on $f+\mb{Z}$. For $n\ge 2$ the state $n+f$ moves to $n+1+f$ with probability $p_\sigma(n+f)$ and to $n-1+f$ otherwise. For $n\ge 1$ the state $-n+f$ moves to $-n-1+f$ with probability $p_\sigma(n-f)$ and to $-n+1+f$ otherwise. The state $f$ moves to $1+f$ with probability $p_\sigma(f)$, to state $-1+f$ with probability $p_\sigma(-f)$, and moves to $2+f$ with probability $r_\sigma(f)$. Finally, for $n = 1$ the state $1+f$ moves to $2+f$ with probability $p_\sigma(1+f) - r_\sigma(f)\exp((2f+1)/(2\sigma^2))$ and to $f$ otherwise.
\end{definition}

\begin{lemma}\label{lem:validity-II}
For $\sigma\ge 1$ and $f\in[-1/2,1/2]$, we have that the walk in \cref{def:ramanujan-walk} is well-defined, and that $\mc{N}(0,\sigma^2)|_{f+\mb{Z}}$ is stationary under a step of the walk with parameters $\sigma,f$.
\end{lemma}
\begin{proof}
That all probabilities are valid follows from \cref{lem:validity}, except that we need to additionally verify
\[p_\sigma(1+f)\ge r_\sigma(f)\exp\bigg(\frac{2f+1}{2\sigma^2}\bigg).\]
This is precisely \cref{lem:inequality}.

To verify that $\mc{N}(0,\sigma^2)|_{f+\mb{Z}}$ is preserved under the walk defined in \cref{def:ramanujan-walk}, recall that $\mc{N}(0,\sigma^2)|_{f+\mb{Z}}$ is preserved under walk defined in \cref{def:jacobi-walk} by \cref{lem:validity}. This walk only differs in its probabilities that $f$ goes to $f,2+f$ and that $1+f$ goes to $f,2+f$. Therefore the probabilities at $n+f$ for $n\in\mb{Z}\setminus\{0,2\}$ are correct. Since the probabilities sum to $1$, it is enough to check the probability at $2+f$ is correct. It therefore suffices to show that
\begin{align*}
r_\sigma(f)\exp\bigg(-&\frac{f^2}{2\sigma^2}\bigg) + \bigg(p_\sigma(1+f)-r_\sigma(f)\exp\bigg(\frac{2f+1}{2\sigma^2}\bigg)\bigg)\exp\bigg(-\frac{(1+f)^2}{2\sigma^2}\bigg)\\
&+ (1-p_\sigma(3+f))\exp\bigg(-\frac{(3+f)^2}{2\sigma^2}\bigg)= \exp\bigg(-\frac{(2+f)^2}{2\sigma^2}\bigg).
\end{align*}
We already verified in the proof of \cref{lem:validity} that
\[p_\sigma(x-1)\exp\bigg(-\frac{(x-1)^2}{2\sigma^2}\bigg) + (1-p_\sigma(x+1))\exp\bigg(-\frac{(x+1)^2}{2\sigma^2}\bigg) = \exp\bigg(-\frac{x^2}{2\sigma^2}\bigg).\]
Plugging in $x = 2+f$ gives the desired identity, upon canceling the terms containing $r_\sigma(f)$.
\end{proof}

Again, we immediately derive a walk which preserves $\mc{N}(0,\sigma^2)$ by piecing together all $f\in[-1/2,1/2)$. Let $R_x^\sigma$ be the random variable defined by writing $x = n+f$, where $f\in[-1/2,1/2)$, and then performing a step according to \cref{def:jacobi-walk}.
\begin{lemma}\label{lem:gaussian-ramanujan}
If $\sigma\ge 1$ and $Z = \mc{N}(0,\sigma^2)$ then $Z + R_Z^\sigma$ is distributed as $\mc{N}(0,\sigma^2)$.
\end{lemma}

\section{Algorithmic Applications}\label{sec:algorithms}
We now derive a number of algorithmic consequences. 

\begin{algorithm}[ht]
\caption{$\tsc{PartialColoring}_\sigma(v_1,\cdots,v_t)$ \label{alg:partial-coloring}}
$w_0 \leftarrow \mc{N}(0,\sigma^2I_n)$ \\
\For{$1\le i\le t$}{
    $\sigma'\leftarrow\sigma/\snorm{v_i}_2$\\
    $x'\leftarrow\sang{w_{i-1},v_i}/\snorm{v_i}_2$\\
    $w_i \leftarrow w_{i-1} +  J_{x'}^{\sigma'}v_i.$ \label{line:move-partial}
}
$w\leftarrow w_t - w_0$
\end{algorithm}

\begin{algorithm}[ht]\label{alg:balance}
\caption{$\tsc{Balancing}_\sigma(v_1,\cdots,v_t)$}
$w_0 \leftarrow \mc{N}(0,\sigma^2I_n)$ \\
\For{$1\le i\le t$}{
    $\sigma'\leftarrow\sigma/\snorm{v_i}_2$\\
    $x'\leftarrow\sang{w_{i-1},v_i}/\snorm{v_i}_2$\\
    $w_i \leftarrow w_{i-1} +  R_{x'}^{\sigma'}v_i.$ \label{line:move-balance}
}
$w\leftarrow w_t - w_0$
\end{algorithm}
In both $\tsc{Balacing}_\sigma$ and  $\tsc{PartialColoring}_\sigma$, $J$ and $R$ are sampled independently every time. Additionally, note that $\tsc{Balacing}_\sigma$ is only well-defined when $\sigma\ge 1$. Finally, we clearly see that $\tsc{PartialColoring}_\sigma$ assigns a sign of $\pm 1$ to each given vector online, or chooses to omit it (a sign of $0$), while $\tsc{Balancing}_\sigma$ does the same except that the sign $2$ is the additional alternative.

Our first algorithm application is a (weak version) of the partial coloring lemma.
\begin{theorem}\label{thm:partial-coloring}
Let $\snorm{v_1}_2,\ldots,\snorm{v_t}_2\le 1$ and $\delta\in(0,1/2)$. With probability at least $1-\delta$ we have that $w_\ell-w_0$ in $\tsc{PartialColoring}_1(v_1,\ldots,v_t)$ is $2\sqrt{2\log(2nt/\delta)}$-bounded for all times $\ell\in[t]$. Furthermore, with probability at least $1-\delta$ we have that $w_t-w_0$ is $2\sqrt{2\log(2n/\delta)}$-bounded. Finally, at least $96.3\%$ of vectors are used with probability $1-\exp(-\Omega(t))$.
\end{theorem}
\begin{proof}
By \cref{lem:gaussian-jacobi} we immediately see that $w_i\sim\mc{N}(0,\sigma^2I_n)$ for all $i\in[t]$. The discrepancy results follow by trivial Gaussian estimates. For example, we see that the $j$th coordinate of $w_\ell$ is $\sqrt{2\log(2nt/\delta)}$-bounded with probability at least $\delta/(2nt)$. Taking a union bound over $0\le\ell\le t$ and $j\in[n]$ yields that $w_0,\ldots,w_t$ are bounded with probability at least $1-\delta$. Therefore each difference is also bounded.

The fraction of vectors used being large follows from Chernoff's inequality and the fact that at every step, conditional on all previous choices, a vector is used with probability at least
\[\min_{f\in[-1/2,1/2]}(1-r_1(f))\ge 0.9639.\qedhere\]
\end{proof}

Our second algorithmic application recovers the online vector balancing results of Alweiss, the first author, and the third author \cite[Theorems~1.1,~1.2]{ALS20}.
\begin{theorem}\label{thm:full-coloring}
Let $\snorm{v_1}_2,\ldots,\snorm{v_t}_2\le 1$, $\delta\in(0,1/2)$, and set $\sigma = \sqrt{\log(t/\delta)}$. With probability at least $1-\delta$ we have that $w_\ell-w_0$ in $\tsc{PartialColoring}_\sigma(v_1,\cdots,v_t)$ is $2\sqrt{2\log(t/\delta)\log(2nt/\delta)}$-bounded for all times $\ell\in [t]$. Furthermore, with probability at least $1-\delta$ we have that $w_t-w_0$ is $2\sqrt{2\log(t/\delta)\log(2n/\delta)}$-bounded. Finally, all vectors are used with probability at least $1-\delta$.
\end{theorem}
\begin{proof}
The proof is essentially identical to that of \cref{thm:partial-coloring}. The only difference is that we see that at each step, a vector is not used with probability at most
\[\max_{f\in[-1/2,1/2]}r_\sigma(f)\le e^{-\sigma^2} = \frac{\delta}{t}\]
due to our choice of $\sigma$, by the inequality in \cref{lem:validity}. A union bound shows that all vectors are used with probability at least $1-\delta$.
\end{proof}
In fact, we can design an algorithm achieving the same bounds by using Algorithm \ref{alg:partial-coloring} for any value of $\sigma \ge 1$ as follows. To do this, first run Algorithm \ref{alg:partial-coloring}, and then rerun Algorithm \ref{alg:partial-coloring} on the vectors which were given a $0$ sign until no vectors remain (note that this can still be done in an online manner). By \cref{lem:validity}, specifically $r_\sigma(f) \le e^{-\sigma^2},$ this process will terminate with probability $1-\delta$ in $O(\sigma^{-2} \log(t/\delta))$ rounds. Each run produces a random vector with variance $O(\sigma^2)$ in every coordinate, hence the total variance is $O(\log (t/\delta))$ per coordinate as desired.

Finally we recover an online version of Banaszczyk \cite{Ban98}, except using $\pm 1, 2$-signings. The proof is identical to that of \cref{thm:partial-coloring} so we omit it.
\begin{theorem}\label{thm:balancing}
Let $\snorm{v_1}_2,\ldots,\snorm{v_t}_2\le 1$ and $\delta\in(0,1/2)$. With probability at least $1-\delta$ we have that $w_\ell-w_0$ in $\tsc{Balancing}_1(v_1,\ldots,v_t)$ is $2\sqrt{2\log(2nt/\delta)}$-bounded for all times $\ell\in[t]$. Furthermore, with probability at least $1-\delta$ we have that $w_t-w_0$ is $2\sqrt{2\log(2n/\delta)}$-bounded.
\end{theorem}

All three algorithmic procedures are online.

\subsection{Computational details}\label{sub:computational-details}
In the previous section the above idealized algorithms ignored the cost of computing $r_\sigma(f)$ and $p_\sigma(n\pm f)$ to sufficient precision in order to be used for algorithmic purposes. The key claim is that one can approximate the above sums within $\delta$ in $\on{poly}(\log(\sigma/\delta))$-time.

In order to do so first note that we can truncate the sums $p_\sigma(n\pm f)$ and $r_\sigma(f)$ to values of $j\ge 1$ where $(j^2+2(n\pm f)j)/(2\sigma^2) = O(\log(\sigma/\delta))$. We now note that
\[\bigg|e^x-\sum_{j=0}^m\frac{x^j}{j!}\bigg|\le\frac{x^{m+1}}{(m+1)!}e^{\max(0,x)},\]
so taking $m = \Theta(\log(\sigma/\delta))$ gives a very good approximation to $\exp(-(j^2+2(n\pm f)j)/(2\sigma^2))$ in the range of terms considered. Now we can compute the desired sums by interpreting it as a sum of low degree (i.e. $O(\log(\sigma/\delta))$) polynomials on a sequence of integers, which can be evaluated quickly.

In the implementation of the algorithms above, at time $t$ if we are given a vector shorter than $1/(2t^2)$, we deterministically add it but ignore it for the purposes of maintaining a Gaussian distribution. These vectors have total length at most $1$, so contribute only $O(1)$ discrepancy in each coordinate. For the remaining vectors, we have $\sigma\le 2t^2$. We thus can approximate the relevant probabilities to within $\delta/(2t^2)$ efficiently, and then sample appropriately. This will preserve the Gaussians in question up to total variation distance of at most
\[\sum_{t\ge 1}\frac{\delta}{2t^2}\le\delta.\]
Therefore, the running time of all probability computations is $\on{poly}(\log (t/\delta))$ at time $t$. Thus the modified versions of the algorithms in \cref{thm:partial-coloring,thm:balancing,thm:full-coloring} run in $O\left(t\on{poly}(\log(t/\delta))+n+\on{nnz}(\{v_i\}_{i\in [t]})\right)$ time with discrepancy guarantees that are an absolute multiplicative factor worse. (The second term arises due to sampling the initial Gaussian point.) This running time essentially matches (up to logarithmic factors) the results of \cite{ALS20} and make progress towards input-sparsity time algorithms for discrepancy, a direction suggested by \cite{DadWeb}.

A variant of our algorithms which run in $O\left(t\on{poly}(\log(t/\delta))+n\log t+\on{nnz}(\{v_i\}_{i\in [t]})\right)$ time is achieved by ``disregarding vectors'' at time $t$ which are shorter than $1/(2t^2)$ (as above) and otherwise grouping vectors by length into dyadic scales and running the algorithms separately with independent randomness on each of the scales. Note that when vector lengths are forced to live in a dyadic scale then sampling an appropriate Gaussian leads us to compute the above probabilities only when $\sigma\in[1,2]$ and hence directly evaluation of the series is efficient.

\section*{Acknowledgements}
The authors thank Ryan Alweiss, Arun Jambulapati, Allen Liu, and Mark Sellke for earlier discussions on this topic. Furthermore we thank Ghaith Hiary for discussions regarding computing theta series.

\bibliographystyle{amsplain0.bst}
\bibliography{main.bib}

\providecommand{\bysame}{\leavevmode\hbox to3em{\hrulefill}\thinspace}
\providecommand{\MR}{\relax\ifhmode\unskip\space\fi MR }
\providecommand{\MRhref}[2]{%
  \href{http://www.ams.org/mathscinet-getitem?mr=#1}{#2}
}
\providecommand{\href}[2]{#2}
\begin{thebibliography}{10}

\bibitem{ALS20}
Ryan Alweiss, Yang~P Liu, and Mehtaab Sawhney, \emph{Discrepancy minimization
  via a self-balancing walk}, arXiv:2006.14009.

\bibitem{And65}
George~E. Andrews, \emph{A simple proof of {J}acobi's triple product identity},
  Proc. Amer. Math. Soc. \textbf{16} (1965), 333--334.

\bibitem{ANSV18}
Juhan Aru, Bhargav Narayanan, Alex Scott, and Ramarathnam Venkatesan,
  \emph{Balancing sums of random vectors}, Discrete Anal. (2018), Paper No. 4,
  17.

\bibitem{Ban98}
Wojciech Banaszczyk, \emph{Balancing vectors and {G}aussian measures of
  n-dimensional convex bodies}, Random Structures \& Algorithms \textbf{12}
  (1998), 351--360.

\bibitem{Ban10}
Nikhil Bansal, \emph{Constructive algorithms for discrepancy minimization},
  51th Annual {IEEE} Symposium on Foundations of Computer Science, {FOCS} 2010,
  October 23-26, 2010, Las Vegas, Nevada, {USA}, {IEEE} Computer Society, 2010,
  pp.~3--10.

\bibitem{BDG16}
Nikhil Bansal, Daniel Dadush, and Shashwat Garg, \emph{An algorithm for
  {K}oml{\'{o}}s conjecture matching {B}anaszczyk's bound}, {IEEE} 57th Annual
  Symposium on Foundations of Computer Science, {FOCS} 2016, 9-11 October 2016,
  Hyatt Regency, New Brunswick, New Jersey, {USA} (Irit Dinur, ed.), {IEEE}
  Computer Society, 2016, pp.~788--799.

\bibitem{BDGL18}
Nikhil Bansal, Daniel Dadush, Shashwat Garg, and Shachar Lovett, \emph{The
  gram-schmidt walk: a cure for the {B}anaszczyk blues}, Proceedings of the
  50th Annual {ACM} {SIGACT} Symposium on Theory of Computing, {STOC} 2018, Los
  Angeles, CA, USA, June 25-29, 2018 (Ilias Diakonikolas, David Kempe, and
  Monika Henzinger, eds.), {ACM}, 2018, pp.~587--597.

\bibitem{BG17}
Nikhil Bansal and Shashwat Garg, \emph{Algorithmic discrepancy beyond partial
  coloring}, Proceedings of the 49th Annual ACM SIGACT Symposium on Theory of
  Computing, 2017, pp.~914--926.

\bibitem{BJMSS20}
Nikhil Bansal, Haotian Jiang, Raghu Meka, Sahil Singla, and Makrand Sinha,
  \emph{Online discrepancy minimization for stochastic arrivals}, Proceedings
  of the 2021 {ACM-SIAM} Symposium on Discrete Algorithms, {SODA} 2021, Virtual
  Conference, January 10 - 13, 2021 (D{\'{a}}niel Marx, ed.), {SIAM}, 2021,
  pp.~2842--2861.

\bibitem{BJSS20}
Nikhil Bansal, Haotian Jiang, Sahil Singla, and Makrand Sinha, \emph{Online
  vector balancing and geometric discrepancy}, Proceedings of the 52nd Annual
  ACM SIGACT Symposium on Theory of Computing (New York, NY, USA), STOC 2020,
  Association for Computing Machinery, 2020, p.~1139–1152.

\bibitem{BS19}
Nikhil Bansal and Joel~H Spencer, \emph{On-{L}ine {B}alancing of {R}andom
  {I}nputs}, arXiv:1903.06898.

\bibitem{DadWeb}
Daniel Dadush,
  \url{https://homepages.cwi.nl/~dadush/workshop/discrepancy-ip/open-problems.html}.

\bibitem{DNTT18}
Daniel Dadush, Aleksandar Nikolov, Kunal Talwar, and Nicole Tomczak-Jaegermann,
  \emph{Balancing vectors in any norm}, 2018 IEEE 59th Annual Symposium on
  Foundations of Computer Science (FOCS), IEEE, 2018, pp.~1--10.

\bibitem{ES18}
Ronen Eldan and Mohit Singh, \emph{Efficient algorithms for discrepancy
  minimization in convex sets}, Random Struct. Algorithms \textbf{53} (2018),
  289--307.

\bibitem{LM15}
Shachar Lovett and Raghu Meka, \emph{Constructive discrepancy minimization by
  walking on the edges}, {SIAM} J. Comput. \textbf{44} (2015), 1573--1582.

\bibitem{Rot17}
Thomas Rothvoss, \emph{Constructive discrepancy minimization for convex sets},
  SIAM Journal on Computing \textbf{46} (2017), 224--234.

\bibitem{Spe85}
Joel Spencer, \emph{Six standard deviations suffice}, Trans. Amer. Math. Soc.
  \textbf{289} (1985), 679--706.

\end{thebibliography}
\appendix
\section{Proof of \texorpdfstring{\cref{lem:inequality}}{Lemma 3.1}}\label{sec:appendix}
\begin{proof}[Proof of \cref{lem:inequality}]
First note $r_\sigma(f)\le r_\sigma(0)$ and $r_\sigma(0)\le r_1(0)$ follow immediately from the Jacobi triple product identity \cref{eq:jacobi-triple-product} and nonnegativity. Therefore it suffices to prove that
\[p_\sigma(1+f)\ge\exp(1/\sigma^2)r_1(0)\]
for all $\sigma\ge 1$ and $f\in[-1/2,1/2]$.

First suppose that $\sigma\in[1,2]$. Then since $p_\sigma(1+f)$ is an alternating series with decreasing terms,
\[p_\sigma(1+f)\ge\exp\bigg(-\frac{1+2(1+f)}{2\sigma^2}\bigg)-\exp\bigg(-\frac{2+2(1+f)}{\sigma^2}\bigg).\]
Fixing $\sigma$, the right has derivative
\[-\frac{1}{\sigma^2}\exp\bigg(-\frac{1+2(1+f)}{2\sigma^2}\bigg)+\frac{2}{\sigma^2}\exp\bigg(-\frac{2+2(1+f)}{\sigma^2}\bigg),\]
which we can check is positive for $f$ underneath some cutoff and negative above this cutoff. Therefore the earlier expression is minimized over $f\in[-1/2,1/2]$ at some $f\in\{\pm 1/2\}$. Then, numerical checking shows that for each case $f\in\{\pm 1/2\}$ the resulting expression is minimized on $\sigma\in\{1,2\}$ for similar reasons. We find the true minimum is at $f = 1/2$ and $\sigma = 1$, which gives
\[p_\sigma(1+f)\ge 0.12\ge er_1(0)\ge\exp(1/\sigma^2)r_1(0).\]

Now we suppose that $\sigma\ge 2$. Let $2k-1$ be the smallest odd integer larger than $\sigma-1-f$, which is clearly always a positive integer as $\sigma\ge 1$ and $f\le 1/2$. We know that $t\mapsto\exp(-t^2/(2\sigma^2))$ is convex for $t\ge\sigma$, hence $t\mapsto\exp(-(t^2+2(1+f)t)/(2\sigma^2))$ is certainly convex and decreasing for $t\ge\sigma-1-f$. Therefore the difference between the values at $j$ and $j+1$ is at least the difference between the values at $j+1$ and $j+2$ when $j\ge 2k-1$, yielding
\begin{align*}
p_\sigma(1+f) &= \sum_{j\ge 1}(-1)^{j-1}\exp\bigg(-\frac{j^2+2(1+f)j}{2\sigma^2}\bigg)\\
&\ge\sum_{j\ge 2k-1}(-1)^{j-1}\exp\bigg(-\frac{j^2+2(1+f)j}{2\sigma^2}\bigg)\\
&\ge\frac{1}{2}\bigg(\sum_{j\ge 2k-1}(-1)^{j-1}\exp\bigg(-\frac{j^2+2(1+f)j}{2\sigma^2}\bigg)+\sum_{j\ge 2k}(-1)^j\exp\bigg(-\frac{j^2+2(1+f)j}{2\sigma^2}\bigg)\bigg)\\
&= \frac{1}{2}\exp\bigg(-\frac{(2k-1)^2+2(1+f)(2k-1)}{2\sigma^2}\bigg)\\
&\ge\frac{1}{2}\exp\bigg(-\frac{(\sigma+1-f)^2+2(1+f)(\sigma+1-f)}{2\sigma^2}\bigg)\\
&\ge\frac{1}{2}\exp\bigg(-\frac{4\sigma^2+16\sigma+15}{8\sigma^2}\bigg)\ge\frac{1}{2}\exp(-71/32)\exp(1/\sigma^2)\\
&\ge 0.05\exp(1/\sigma^2)\ge\exp(1/\sigma^2)r_1(0).\qedhere
\end{align*}
\end{proof}

\end{document}